\documentclass{comsoc2010}

\usepackage[amsthm]{pamas}

\usepackage{pgflibraryshapes}  
\usetikzlibrary{positioning}  
\usetikzlibrary{patterns,arrows,shapes}

\usepackage{algorithm,algorithmic} 

\newcommand{\pref}{\mathcal{P}\xspace}

\newcommand{\ASHG}{ASHG\xspace}
\newcommand{\ASHGS}{ASHGs\xspace}

\newcommand\eat[1]{}

\title{Stable partitions in additively separable hedonic games\footnote{This material is based on work supported by the Deutsche Forschungsgemeinschaft under grants BR-2312/6-1 (within the European Science Foundation's EUROCORES program LogICCC) and BR~2312/7-1.}}
\author{Haris Aziz, Felix Brandt, and Hans Georg Seedig}

\begin{document}

	\begin{abstract}
	An important aspect in systems of multiple autonomous agents is the exploitation of synergies via coalition formation.
	In this paper, we solve various open problems concerning the computational complexity of stable partitions in additively separable hedonic games. 
	First, we propose a polynomial-time algorithm to compute a contractually individually stable partition. This contrasts with previous results such as the NP-hardness of computing individually stable or Nash stable partitions. 
	Secondly, we prove that checking whether the core or the strict core exists is NP-hard in the strong sense even if the preferences of the players are symmetric. Finally, it is shown that verifying whether a partition consisting of the grand coalition is contractually strict core stable or Pareto optimal is coNP-complete.
	\end{abstract}

	\section{Introduction}

	Ever since the publication of \citeauthor{vNM47a}'s \emph{Theory of Games and Economic Behavior} in 1944, coalitions have played a central role within game theory. The crucial questions in coalitional game theory are
	which coalitions can be expected to form and how the members of coalitions should divide the proceeds of their cooperation. Traditionally the focus has been on the latter issue, which led to the formulation and analysis of concepts such as the core, the Shapley value, or the bargaining set.
	Which coalitions are likely to form is commonly assumed to be settled exogenously, either by explicitly specifying the coalition structure, a partition of the players in disjoint coalitions, or, implicitly, by assuming that larger coalitions can invariably guarantee better outcomes to its members than smaller ones and that, as a consequence, the grand coalition of all players will eventually form. 
	The two questions, however, are clearly interdependent: the individual players' payoffs depend on the coalitions that form just as much as the formation of coalitions depends on how the payoffs are distributed.

	\emph{Coalition formation games}, as introduced by \citet{DrGr80a}, provide a simple but versatile formal model that allows one to focus on coalition formation. In many situations it is natural to assume that a player's appreciation of a coalition structure only depends on the coalition he is a member of and not on how the remaining players are grouped.
	Initiated by  \citet{BKS01a} and \citet{BoJa02a}, much of the work on coalition formation now concentrates on these so-called \emph{hedonic games}. Hedonic games are relevant in modeling many settings such as formation of groups, clubs and societies~\citep{BoJa02a} and also online social networking~\citep{ElkiW09a}.

	The main focus in hedonic games has been on notions of stability for coalition structures such as \emph{Nash stability}, \emph{individual stability}, \emph{contractual individual stability}, or \emph{core stability} and characterizing conditions under which the set of stable partitions is guaranteed to be non-empty \citep[see, \eg][]{BoJa02a,BuZw03a}. \citet{SuDi07b} presented a taxonomy of stability concepts which includes the \emph{contractual strict core}, the most general stability concept that is guaranteed to exist. 
	A well-studied special case of hedonic games are two-sided matching games in which only coalitions of size two are admissible~\citep{RoSo90a}. 
	We refer to~\citet{Hajd06a} for a critical overview of hedonic games.

	Hedonic games have recently been examined from an algorithmic perspective~\citep[see, \eg][]{Ball04a,DBHS06a}. \citet{Cech08a} surveyed the algorithmic problems related to stable partitions in hedonic games in various representations. 
	\citet{Ball04a} showed that for hedonic games represented by \emph{individually rational list of coalitions}, the complexity of checking whether core stable, Nash stable, or individual stable partitions exist is NP-complete. He also proved that every hedonic game admits a contractually individually stable partition. 
	Coalition formation games have also received attention in the artificial intelligence community where the focus has generally been on computing optimal partitions for general coalition formation games~\cite{SLA+99a} without any combinatorial structure. 
	\citet{ElkiW09a} proposed a fully-expressive model to represent hedonic games which encapsulates well-known representations such as \emph{individually rational list of coalitions} and \emph{additive separability}.

	\emph{Additively separable hedonic games (ASHGs)} constitute a particularly natural and succinctly representable class of hedonic games. Each player in an ASHG has a value for any other player and the value of a coalition to a particular player is simply the sum of the values he assigns to the members of his coalition. Additive separability satisfies a number of desirable axiomatic properties~\citep{BBP04a} and \ASHGS are the non-transferable utility generalization of \emph{graph games} studied by \citet{DePa94a}. 
	\citet{Olsen09a} showed that checking whether a nontrivial Nash stable partition exists in an \ASHG is NP-complete if preferences are nonnegative and symmetric. This result was improved by \citet{SuDi10a} who showed that checking whether a core stable, strict core stable, Nash stable, or individually stable partition exists in a general \ASHG is NP-hard.
	\citet{DBHS06a} obtained positive algorithmic results for subclasses of \ASHGS in which each player merely divides other players into friends and enemies. \citet{BrLa09a} examined the tradeoff between stability and social welfare in \ASHGS.
	Recently, \citet{GaSa10a} showed that computing partitions that satisfy some variants of individual-based stability is PLS-complete, even for very restricted preferences. 
	In another paper, 
	\citet{ABS10a} studied the complexity of computing and verifying optimal partitions in \ASHGS.

	In this paper, we settle the complexity of key problems regarding stable partitions of \ASHGS. 
	We present a polynomial-time algorithm to compute a contractually individually stable partition. This is the first positive algorithmic result (with respect to one of the standard stability concepts put forward by \citet{BoJa02a}) for general \ASHGS with no restrictions on the preferences. We strengthen recent results of \citet{SuDi10a} and prove that checking whether the core or the strict core exists is NP-hard, even if the preferences of the players are symmetric. Finally, it is shown that verifying whether a partition is in the contractually strict core (CSC) is coNP-complete, even if the partition under question consists of the grand coalition. This is the first computational hardness result concerning CSC stability in hedonic games of any representation. The proof can be used to show that verifying whether the partition consisting of the grand coalition is Pareto optimal is coNP-complete, thereby answering a question mentioned by \citet{ABS10a}. Our computational hardness results imply computational hardness of the equivalent questions for \emph{hedonic coalition nets} \citep{ElkiW09a}.

	\section{Preliminaries}
	\label{sec:pre}

	In this section, we provide the terminology and notation required for our results.

	A \emph{hedonic coalition formation game} is a pair $(N,\pref)$ where $N$ is a set of players and $\pref$ is a \emph{preference profile} which specifies for each player $i\in N$ the preference relation $ \succsim_i$, a reflexive, complete, and transitive binary relation on the set $\mathcal{N}_i=\{S\subseteq N \mid i\in S\}$.
	The statement $S\succ_iT$ denotes that $i$ strictly prefers $S$ over $T$ whereas $S\sim_iT$ means that $i$ is indifferent between coalitions $S$
	 and $T$. A \emph{partition} $\pi$ is a partition of players $N$ into disjoint coalitions. By $\pi(i)$, we denote the coalition of $\pi$ that includes player $i$.

	We consider utility-based models rather than purely ordinal models. 
	In \emph{additively separable preferences}, a player $i$ gets value $v_i(j)$ for player $j$ being in the same coalition as $i$ and if $i$ is in coalition $S\in \mathcal{N}_i$, then $i$ gets utility $\sum_{j\in S\setminus \{i\}}v_i(j)$.
	A game $(N,\pref)$ is \emph{additively separable} if for each player $i\in N$, there is a utility function $v_i: N\rightarrow \mathbb{R}$ such that $v_i(i)=0$ and for coalitions $S,T\in\mathcal{N}_i$, $S \succsim_i T$ if and only if $\sum_{j\in S}v_i(j) \geq \sum_{j\in T}v_i(j)$.  We will denote the utility of player $i$ in partition $\pi$ by $u_{\pi}(i)$.

	A preference profile is \emph{symmetric} if $v_i(j)=v_j(i)$ for any two players $i,j\in N$ and is \emph{strict} if $v_i(j)\neq 0$ for all $i,j\in N$. 
	For any player $i$, let $F(i,A)=\{j\in A \mid v_i(j)> 0\}$ be the set of friends of player $i$ within $A$.

	We now define important stability concepts used in the context of coalition formation games.

	\renewcommand*{\labelitemi}{$\bullet$}

	\begin{itemize}
	\item A partition is  \emph{Nash stable (NS)} if no player can benefit by 
	moving from his coalition $S$ to another (possibly empty) coalition $T$.
	\item A partition is \emph{individually stable (IS)} if no player can
	benefit by moving from his coalition $S$ to another existing (possibly empty) coalition $T$  while not making the members of $T$ worse off.
	\item A partition is  \emph{contractually individually stable (CIS)} if no
	player can benefit by moving from his coalition $S$ to another existing (possibly empty) coalition
	$T$ while making neither the members of $S$ nor the members of
	$T$ worse off.
	\item We say that a coalition $S \subseteq N$ \emph{strongly blocks} a partition $\pi$, if each
	player $i \in S$ strictly prefers $S$ to his current coalition $\pi(i)$ in
	the partition $\pi$. A partition which admits no blocking coalition is said to be in the \emph{core (C)}.
	\item We say that a coalition $S \subseteq N$ \emph{weakly blocks} a partition $\pi$,
	if each player $i \in S$ weakly prefers $S$ to $\pi(i)$ and there exists at least one player $j \in S$ who strictly prefers $S$ to his
	current coalition $\pi(j)$. A partition which admits no weakly blocking coalition is in the \emph{strict core (SC)}.
	\eat{\item A partition $\pi$ is in the \emph{contractual strict core (CSC)} if the following is not possible: there exists a weakly blocking coalition $S$ and for all $C\in \pi$, and for all $j\in N\setminus S$ we have that $C\setminus S \succsim_j C$.}
	\item A partition $\pi$ is in the \emph{contractual strict core (CSC)} if any weakly blocking coalition $S$ makes at least one player $j\in N\setminus S$ worse off when breaking off.

	\end{itemize}

	The inclusion relationships between stability concepts depicted in Figure~\ref{fig:relations} follow from the definitions of the concepts. 
	We will also consider \emph{Pareto optimality}. 
	A partition $\pi$ of $N$ is \emph{Pareto optimal} if there exists no partition $\pi'$ of $N$ such that for all $i\in N$,  $\pi'(i) \succsim_i  \pi(i)$ and there exists at least one player $j\in N$ such that $\pi'(j) \succ_j  \pi(j)$. We say that a partition $\pi$ satisfies \emph{individual rationality} if each player does as well as by being alone, i.e., for all $i\in N$, $\pi(i) \succsim_i  \{i\}$.

	Throughout the paper, we assume familiarity with basic concepts of computational complexity~\citep[see, \eg][]{ArBa09a}.

	\eat{
	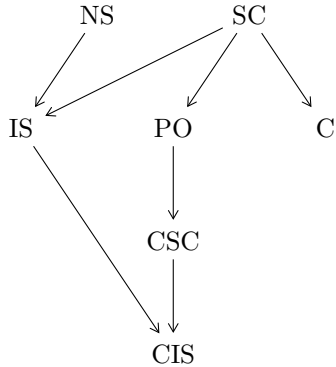
\begin{figure}
	\begin{center}
		\scalebox{1}{
		\begin{tikzpicture}
			\tikzstyle{pfeil}=[->,>=angle 60, shorten >=1pt,draw]
			\tikzstyle{onlytext}=[]

			\node[onlytext] (NS) at (1,0) {NS};
			\node[onlytext] (SC) at (3,0) {SC};
			\node[onlytext] (IS) at (0,-1.5) {IS};
			\node[onlytext] (CSC) at (2,-1.5) {CSC};
			\node[onlytext] (C) at (4,-1.5) {C};
			\node[onlytext] (CIS) at (2,-3) {CIS};

			\draw[pfeil] (NS) to (IS);
			\draw[pfeil] (SC) to (IS);
			\draw[pfeil] (SC) to (CSC);
			\draw[pfeil] (IS) to (CIS);
			\draw[pfeil] (SC) to (C);
			\draw[pfeil] (CSC) to (CIS);
		\end{tikzpicture}
		}
		\end{center}
		\caption{Inclusion relationships between stability concepts. For example, every Nash stable partition is also individually stable.}
		\label{fig:relations}
		\end{figure}

		}

		\begin{figure}
		\begin{center}
			\scalebox{1}{
			\begin{tikzpicture}
				\tikzstyle{pfeil}=[->,>=angle 60, shorten >=1pt,draw]
				\tikzstyle{onlytext}=[]

				\node[onlytext] (NS) at (1,0) {NS};
				\node[onlytext] (SC) at (3,0) {SC};
				\node[onlytext] (IS) at (0,-1.5) {IS};
				\node[onlytext] (PO) at (2,-1.5) {PO};
				\node[onlytext] (C) at (4,-1.5) {C};
				\node[onlytext] (CSC) at (2,-3) {CSC};
				\node[onlytext] (CIS) at (2,-4.5) {CIS};

				\draw[pfeil] (NS) to (IS);
				\draw[pfeil] (SC) to (IS);
				\draw[pfeil] (SC) to (PO);
				\draw[pfeil] (IS) to (CIS);
				\draw[pfeil] (SC) to (C);
				\draw[pfeil] (PO) to (CSC);
					\draw[pfeil] (CSC) to (CIS);
			\end{tikzpicture}
			}

			\end{center}
			\caption{Inclusion relationships between stability concepts. For example, every Nash stable partition is also individually stable.}
			\label{fig:relations}
			\end{figure}

	\eat{	

	\begin{fact}
	The following statements follow from the definitions of the stability concepts.
		\begin{enumerate}
			\item Nash stability $\Longrightarrow$ individual stability $\Longrightarrow$ contractual individual stability
			\item strict core stability $\Longrightarrow$ core stability
			\item strict core stability $\Longrightarrow$ individual stability
			\item contractual strict core stability  $\Longrightarrow$ contractual individual stability
		\end{enumerate}
	\label{fact:implies-stabilty} 
	\end{fact}
	}

	\section{Contractual individual stability}
	\label{sec:cis}

	It is known that computing or even checking the existence of Nash stable or individually stable partitions in an \ASHG is NP-hard. 
	On the other hand, a potential function argument can be used to show that at least one CIS partition exists for every hedonic game~\citep{Ball04a}. 
	The potential function argument does not imply that a CIS partition can be computed in polynomial time. There are many cases in hedonic games, where a solution is guaranteed to exist but \emph{computing} it is not feasible. 
	For example, \citet{BoJa02a} presented a potential function argument for the existence of a Nash stable partition for \ASHGS with symmetric preferences. 
	However there are no known polynomial-time algorithms to \emph{compute} such partitions and there is evidence that there may not be any polynomial-time algorithm~\citep{GaSa10a}. 
	In this section, we show that a CIS partition can be computed in polynomial time for \ASHGS. The algorithm is formally described as Algorithm~\ref{alg-CIS-general}. Algorithm~\ref{alg-CIS-general} may also prove useful as a preprocessing or intermediate routine in other algorithms to compute different types of stable partitions of hedonic games.

	\begin{algorithm}[tb]
	  \caption{CIS partition of an \ASHG}
	  \label{alg-CIS-general}
	  \textbf{Input:} additively separable hedonic game $(N,\pref)$.\\
	  \textbf{Output:} CIS partition.

	  \begin{algorithmic}[] 
	\STATE $i\leftarrow 0$
	\STATE $R\leftarrow N$
	  \WHILE{$R\neq \emptyset$}\label{while-step}
	\STATE Take any player $a\in R$ 
	\STATE $h \leftarrow \sum_{b\in F(a,R)}v_a(b)$
	\STATE $z\leftarrow i+1$

	\FOR{$k\leftarrow 1$ to $i$}
	\STATE $h'\leftarrow \sum_{b\in S_k}v_a(b)$
	 \IF{($h < h'$) ~$\wedge~ (\forall b\in S_k$, $v_b(a)=0$)}
	\STATE $h\leftarrow h'$
	\STATE $z \leftarrow k$
	\ENDIF
	\ENDFOR

	\IF[$a$ is latecomer]{$z\ne i+1$}
	\STATE $S_{z}\leftarrow \{a\}\cup S_{z}$
	\STATE $R\leftarrow R \setminus \{a\}$

	\ELSE[$a$ is leader] 

	\STATE $i\leftarrow z$
	\STATE $S_i\leftarrow \{a\}$
	\STATE $S_i\leftarrow S_i\cup F(a,R)$ \COMMENT{add leader's helpers}
	\STATE $R\leftarrow R \setminus S_i$

	 \ENDIF

	\WHILE{$\exists j\in R$ such that $\forall i\in S_{z}$, $v_i(j)\geq 0$ and $\exists i\in S_{z}$, $v_i(j)>0$} 
	\STATE $R\leftarrow R\setminus \{j\}$
	\STATE $S_{z}\leftarrow S_{z} \cup \{j\}$ \COMMENT{add needed players}
	\ENDWHILE

		\ENDWHILE

	  \RETURN $\{S_1,\ldots, S_i\}$
	 \end{algorithmic}
	\end{algorithm}
	\normalsize

	\begin{theorem}
	A CIS partition can be computed in polynomial time.	
	\label{prop:CIS-easy}
	\end{theorem}
	\begin{proof}
		Our algorithm to compute a CIS partition can be viewed as successively giving a priority token to players to form the best possible coalition among the remaining players or join the best possible coalition which tolerates the player.
		The basic idea of the algorithm is described informally as follows. Set variable $R$ to $N$ and consider an arbitrary player $a\in R$.
		Call $a$ the \emph{leader} of the first coalition $S_i$ with $i=1$.
		Move any player $j$ such that $v_a(j)>0$ from $R$ to $S_i$. 
		Such players are called the \emph{leader's helpers}. 
		Then keep moving any player from $R$ to $S_i$ which is tolerated by all players in $S_i$ and strictly liked by at least one player in $S_i$. 
		Call such players \emph{needed players}. 
		Now increment $i$ and take another player $a$ from among the remaining players $R$ and check the maximum utility he can get from among $R$. 
		If this utility is less than the utility which can be obtained by joining a previously formed coalition in $\{S_1,\ldots, S_{i-1}\}$, then send the player to such a coalition where he can get the maximum utility (as long all players in the coalition tolerate the incoming player). 
		Such players are called \emph{latecomers}.
		Otherwise, form a new coalition $S_i$ around $a$ which is the best possible coalition for player $a$ taking only players from the remaining players $R$.
		Repeat the process until all players have been dealt with and $R=\emptyset$. 
		We prove by induction on the number of coalitions formed that no CIS deviation can occur in the resulting partition. 
		The hypothesis is the following: 

	\emph{Consider the $k$th first formed coalitions $S_1,\ldots, S_k$. Then neither of the following can happen:
	\begin{enumerate}
	\item There is a CIS deviation by a player from among $S_1,\ldots, S_k$.
	\item There is a CIS deviation by a player from among $N\setminus \bigcup_{i\in\{1,\ldots,k\}}S_i$ to a coalition in $\{S_1,\ldots, S_k\}$.
	\end{enumerate}
	}

	\noindent
	\paragraph{Base case} Consider the coalition $S_1$. Then the leader of $S_1$ has no incentive to leave. 
	The leader's helpers are not allowed to leave by the leader. If they did, the leader's utility would decrease.
	For each of the needed players, there exists one player in $S_1$ who does not allow the needed player to leave.  
	Now let us assume a latecomer $i$ arrives in $S_1$. This is only possible if the maximum utility that the latecomer can derive from a coalition $C\subseteq (N\setminus S_1)$ is less than $\sum_{j\in S_1}v_{i}(j)$. Therefore once $i$ joins $S_1$, he will only become less happy by leaving $S_1$.

	Any player $i\in N\setminus S_1$ cannot have a CIS deviation to $S_1$. Either $i$ is disliked by at least one player in $S_1$ or $i$ is disliked by no player in $S_1$. In the first case, $i$ cannot deviate to $S_1$ even he has an incentive to. In the second case, player $i$ has no incentive to move to $S_1$ because if he had an incentive, he would already have moved to $S_1$ as a latecomer.

	\paragraph{Induction step} Assume that the hypothesis is true. Then we prove that the same holds for the formed coalitions $S_1,\ldots, S_k,S_{k+1}$. By the hypothesis, we know that players cannot leave coalitions $S_1,\ldots, S_k$.
	Now consider $S_{k+1}$. The leader $a$ of $S_{k+1}$ is either not allowed to join one of the coalitions in 
	$\{S_1,\ldots,S_k\}$
	or if he is, he has no incentive to join it. Player $a$ would already have been member of $S_i$ for some $i\in \{1,\ldots, k\}$ if one of the following was true:

	\begin{itemize}
		\item There is some $i\in\{1,\dots,k\}$ such that the leader of $S_i$ likes $a$.
		\item  There is some $i\in\{1,\dots,k\}$ such that for all $b\in S_i$, $v_{b}(a)\geq 0$ and there exists $b\in S_i$ such that $v_b(a)>0$.
		\item There is some $i\in\{1,\dots,k\}$, such that for all $b\in S_i$, $v_{b}(a)= 0$ and  $\sum_{b\in S_i}v_{a}(b)> \sum_{b\in F(i,N\setminus \cup_{i=1}^{k}S_i)}v_{a}(b)$ and $\sum_{b\in S_i}v_{a}(b)\geq  \sum_{b\in S_j}v_{a}(b)$ for all $j\in \{1,\ldots, k\}$.
	\end{itemize}

	Therefore $a$ has no incentive or is not allowed to move to another $S_j$ for $j\in \{1,\ldots, k\}$. 
	Also $a$ will have no incentive to move to any coalition formed after $S_1,\ldots, S_{k+1}$ because he can do strictly better in $S_{k+1}$. 
	Similarly, $a$'s helpers are not allowed to leave $S_{k+1}$ even if they have an incentive to. Their movement out of $S_{k+1}$ will cause $a$ to become less happy. 
	Also each needed player in $S_{k+1}$ is not allowed to leave because at least one player in $S_k$ likes him. 
	Now consider a latecomer $l$ in $S_{k+1}$. Latecomer $l$ gets strictly less utility in any coalition 
	$C\subseteq N\setminus \bigcup_{i=1}^{k+1}S_i$. Therefore $l$ has no incentive to leave $S_{k+1}$.

	Finally, we prove that there exists no player $x\in N\setminus \bigcup_{j=1}^{k+1}S_i$ such that $x$ has an incentive to and is allowed to join $S_i$ for $i\in \{1,\ldots k+1\}$. By the hypothesis, we already know that $x$ does not have an incentive or is allowed to a join a coalition $S_i$ for $i\in \{1,\ldots k\}$. Since $x$ is not a latecomer for $S_{k+1}$, $x$ either does not have an incentive to join $S_{k+1}$ or is disliked by at least one player in $S_{k+1}$. 
	\qed
	\end{proof}

	\section{Core and strict core}
	\label{sec:core}

	For \ASHGS, the problem of testing the core membership of a partition is coNP-complete~\citep{SuDi07a}. This fact does not imply that checking the existence of a core stable partition is NP-hard. Recently, \citet{SuDi10a} showed that for \ASHGS checking whether a core stable or strict core stable partition exists is NP-hard in the strong sense. Their reduction relied on the asymmetry of the players' preferences.
	We prove that even with symmetric preferences, checking whether a core stable or a strict core stable partition exists is NP-hard in the strong sense. Symmetry is a natural, but rather strong condition, that can often be exploited algorithmically.

	We first present an example 
	of a six-player \ASHG with symmetric preferences for which the core (and thereby the strict core) is empty.

	\begin{example}
	\label{example:symm-core-empty}
		Consider a six player symmetric \ASHG adapted from an example by \citet{BKS01a} where

		\begin{itemize}
		\item $v_1(2)=v_3(4)=v_5(6)=6$;
		\item $v_1(6)=v_2(3)=v_4(5)=5$;
		\item $v_1(3)=v_3(5)=v_1(5)=4$;
		\item $v_1(4)=v_2(5)=v_3(6)=-33$; and
		\item $v_2(4)=v_2(6)=v_4(6)=-33$
		\end{itemize}
		as depicted in Figure~\ref{fig:example}.

		\begin{figure}
		\centering
		\scalebox{0.9}{
		\begin{tikzpicture}[auto]
			\tikzstyle{player}=[draw, circle, fill=white, minimum size=11pt, inner sep=4pt]
			\tikzstyle{lab}=[font=\small\itshape]

			\draw (0,0) node[player] (x11) {1}
				 ++(90:2cm) node[player] (x12) {2}
				 ++(30:2cm) node[player] (x13) {3}
				 ++(330:2cm) node[player] (x14) {4}
				 ++(270:2cm) node[player] (x15) {5}
				 ++(210:2cm) node[player] (x16) {6};
			\draw[-] (x11) to [bend right = 15] node[lab] {4} (x13);
			\draw[-] (x13) to [bend right = 15] node[lab] {4} (x15);
			\draw[-] (x15) to [bend right = 15] node[lab] {4} (x11);

			\draw[-] (x11) to node[lab] {6} (x12);
			\draw[-] (x12) to node[lab] {5} (x13);
			\draw[-] (x13) to node[lab] {6} (x14);
			\draw[-] (x14) to node[lab] {5} (x15);
			\draw[-] (x15) to node[lab] {6} (x16);
			\draw[-] (x16) to node[lab] {5} (x11);

		\end{tikzpicture}
		}
		\caption{Graphical representation of Example~\ref{example:symm-core-empty}. All edges not shown in the figure have weight $-33$.}
		\label{fig:example}
	\end{figure}
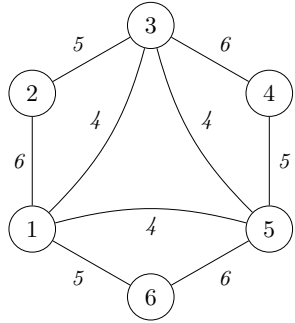

	It can be checked that no partition is core stable for the game. 
	Note that if $v_i(j)=-33$, then $i$ and $j$ cannot be in the same coalition of a core stable partition.
	Also, players can do better than in a partition of singleton players.
	Let coalitions which satisfy individual rationality be called feasible coalitions. We note that the following are the feasible coalitions:
	$\{1,2\}$, $\{1,3\}$, $\{1,5\}$, $\{1,6\}$, $\{1,2,3\}$, $\{1,3,5\}$, $\{1,5,6\}$, $\{2,3\}$, $\{3,4\}$, $\{3,4,5\}$, $\{3,5\}$, $\{4,5\}$ and $\{5,6\}$.

	Consider partition $$\pi=\{\{1,2\}, \{3,4,5\}, \{6\}\}.$$ 

	Then,
	\begin{itemize}
	\item $u_{\pi}(1)=6$;
	\item $u_{\pi}(2)=6$;
	\item $u_{\pi}(3)=10$;
	\item $u_{\pi}(4)=11$;
	\item $u_{\pi}(5)=9$; and
	\item $u_{\pi}(6)=0$.
	\end{itemize}

	Out of the feasible coalitions listed above, the only weakly (and also strongly) blocking coalition is $\{1,5,6\}$ in which player 1 gets utility 9, player 5 gets utility 10, and player 6 gets utility 11. We note that the coalition $\{1,2,3\}$ is not a weakly or strongly blocking coalition because player 3 gets utility 9 in it. Similarly $\{1,3,5\}$ is not a weakly or strongly blocking coalition because both player 3 and player 5 are worse off.
	One way to prevent the deviation $\{1,5,6\}$ is to provide some incentive for player $6$ not to deviate with $1$ and $5$. This idea will be used in the proof of Theorem~\ref{prop:corehard}.

	\end{example}

	We now define a problem that is NP-complete is the strong sense:\\

	\noindent
	\textbf{Name}: {\sc ExactCoverBy3Sets (E3C)}: \\
	\noindent
	\textbf{Instance}: A pair $(R,S)$, where $R$ is a set and $S$ is a collection of subsets of 
	$R$ such that $|R|= 3m$ for some positive integer $m$ and $|s| = 3$ for each 
	$s\in S$. \\
	\noindent
	\textbf{Question}: Is there a sub-collection $S'\subseteq S$ which is a partition of $R$? \\

	It is known that E3C remains NP-complete even if each $r\in R$ occurs in 
	at most three members of $S$~\citep{SuDi10a}. We will use this assumption in the proof of Theorem~\ref{prop:corehard}, which will be shown by a reduction from E3C.

	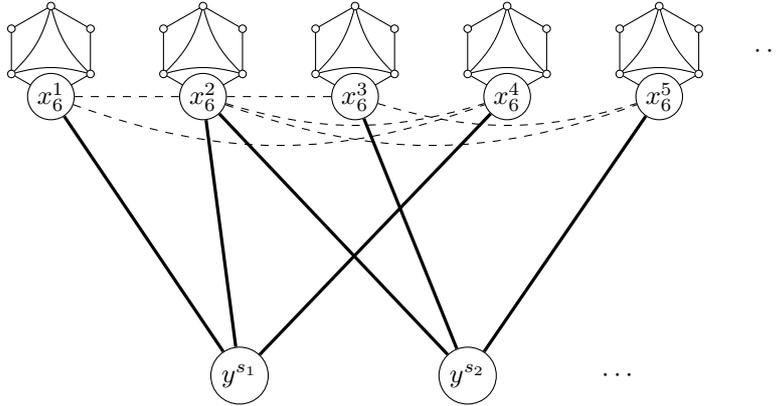
\begin{figure*}[tb]
		\centering
		\scalebox{1}{
		\begin{tikzpicture}[auto]

		\tikzstyle{small}=[draw, circle, fill=white!100,minimum size=2pt, inner sep=1pt]
		\tikzstyle{lab}=[]
		\tikzstyle{collection}=[draw,circle,minimum size=11pt, inner sep=2pt]
		\tikzstyle{dots}=[]

		\foreach \i in {1,2,3,4,5}
		{
			\pgfmathparse{2*\i-2}
			\draw (\pgfmathresult,0) node[small] (x\i1) {}
				-- ++(90:0.6cm) node[small] (x\i2) {}
				-- ++(30:0.6cm) node[small] (x\i3) {}
				-- ++(330:.6cm) node[small] (x\i4) {}
				-- ++(270:.6cm) node[small] (x\i5) {}
				-- ++(210:.6cm) node[small] (x\i6) {$x_6^\i$}
				-- (x\i1);
			\draw[-] (x\i1) to [bend right = 15] node[lab] {} (x\i3);
			\draw[-] (x\i3) to [bend right = 15] node[lab] {} (x\i5);
			\draw[-] (x\i5) to [bend right = 15] node[lab] {} (x\i1);
		}
		\node[dots] (subgame-dots) at (10,0.3) {$\cdots$};

		\node[collection] (y1) at (3,-4) {$y^{s_1}$};
		\node[collection] (y2) at (6,-4) {$y^{s_2}$};
		\node[dots] (collection-dots) at (8,-4) {$\cdots$};

		\draw[-,very thick] (y1) to node {} (x16);
		\draw[-,very thick] (y1) to node {} (x26);
		\draw[-,very thick] (y1) to node {} (x46);

		\draw[-,very thick] (y2) to node {} (x26);
		\draw[-,very thick] (y2) to node {} (x36);
		\draw[-,very thick] (y2) to node {} (x56);	

		\draw[dashed] (x16) to (x26);
		\draw[dashed] (x26) to [bend right = 17] (x46);
		\draw[dashed] (x16) to [bend right = 20] (x46);

		\draw[dashed] (x26) to (x36);
		\draw[dashed] (x36) to [bend right = 17] (x56);
		\draw[dashed] (x26) to [bend right = 20] (x56);

		\end{tikzpicture}

		}

		\caption{Graphical representation of an ASHG derived from an instance of E3C in the proof of Theorem~\ref{prop:corehard}. Symmetric utilities other than $-33$ are given as edges. Thick edges indicate utility $10\frac{1}{4}$ and dashed edges indicate utility $1/2$. Each hexagon at the top looks like the one in Figure~\ref{fig:examplemod}.}
		\label{fig:proofTh2}
	\end{figure*}

	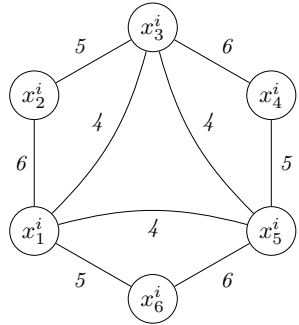
\begin{figure}[tb]
		\centering
		\scalebox{0.9}{
		\begin{tikzpicture}[auto]
			\tikzstyle{player}=[draw, circle, fill=white, minimum size=11pt, inner sep=2pt]
			\tikzstyle{lab}=[font=\small\itshape]
			\draw (0,0) node[player] (x11) {$x_1^i$}
				 ++(90:2cm) node[player] (x12) {$x_2^i$}
				 ++(30:2cm) node[player] (x13) {$x_3^i$}
				 ++(330:2cm) node[player] (x14) {$x_4^i$}
				 ++(270:2cm) node[player] (x15) {$x_5^i$}
				 ++(210:2cm) node[player] (x16) {$x_6^i$};
			\draw[-] (x11) to [bend right = 15] node[lab] {4} (x13);
			\draw[-] (x13) to [bend right = 15] node[lab] {4} (x15);
			\draw[-] (x15) to [bend right = 15] node[lab] {4} (x11);		
			\draw[-] (x11) to node[lab] {6} (x12);
			\draw[-] (x12) to node[lab] {5} (x13);
			\draw[-] (x13) to node[lab] {6} (x14);
			\draw[-] (x14) to node[lab] {5} (x15);
			\draw[-] (x15) to node[lab] {6} (x16);
			\draw[-] (x16) to node[lab] {5} (x11);	
		\end{tikzpicture}
		}
		\caption{Graphical representation of the \ASHG from Example~\ref{example:symm-core-empty} as used in the proof of Theorem~\ref{prop:corehard}. All edges not shown in the figure have weight $-33$.}
		\label{fig:examplemod}
	\end{figure}

	\begin{theorem}
		\label{prop:corehard}
		Checking whether a core stable or a strict core stable partition exists is NP-hard in the strong sense, even when preferences are symmetric.
	\end{theorem}
	\begin{proof}
	Let $(R,S)$ be an instance of E3C where $r\in R$ occurs in at most three members of $S$. We reduce $(R,S)$ to an \ASHGS with symmetric preferences $(N,\pref)$ in which there is a player $y^s$ corresponding to each $s\in S$ and there are six players $x_1^r, \ldots, x_6^r$ corresponding to each $r\in R$. These players have preferences over each other in exactly the way players $1,\ldots, 6$ have preference over each other as in Example~\ref{example:symm-core-empty}. 

	So, $N=\{x_1^r, \ldots, x_6^r \mid r \in R\} \cup \{y ^s \mid s\in S \}$. We assume that all preferences are symmetric. The player preferences are as follows:

	\begin{itemize}
	\item For $i\in R$, \\$v_{x_1^i}(x_2^i)=v_{x_3^i}(x_4^i)=v_{x_5^i}(x_6^i)=6$;\\
	 $v_{x_1^i}(x_6^i)=v_{x_2^i}(x_3^i)=v_{x_4^i}(x_5^i)=5$; and\\
	 $v_{x_1^i}(x_3^i)=v_{x_3^i}(x_5^i)=v_{x_1^i}(x_5^i)=4$;
	\item For any $s=\{k,l,m\}\in S$,\\ $v_{x_6^k}(x_6^l)=v_{x_6^l}(x_6^k)=v_{x_6^k}(x_6^m)=v_{x_6^m}(x_6^k)=v_{x_6^l}(x_6^m)=v_{x_6^m}(x_6^l)=1/2$; and\\
	 $v_{x_6^k}(y^s)=v_{x_6^l}(y^s)=v_{x_6^m}(y^s)=10\frac{1}{4}$;\\
	\item $v_i(j)=-33$ for any $i,j\in N$ for valuations not defined above. 
	\end{itemize}

	We prove that $(N,P)$ has a non-empty strict core (and thereby core) if and only if there exists an $S'\subseteq S$ such that $S'$ is a partition of $R$.

	Assume that there exists an $S'\subseteq S$ such that $S'$ is a partition of $R$. Then we prove that there exists a strict core  stable (and thereby core stable) partition $\pi$ where $\pi$ is defined as follows:
	\begin{eqnarray*}
	&&\{\{x_1^i, x_2^i\}, \{x_3^i,x_4^i,x_5^i\}\mid i\in R\}
	 \cup \{\{y^s\}\mid s\in S\setminus S'\}\\
	& \cup&\{\{y^s\cup \{x_6^i\mid i\in s\}\}\mid s\in S'\}\text{.}
	\end{eqnarray*}

	For all $i\in R$, 

	\begin{itemize}
	\item $u_{\pi}(x_1^i)=6$;
	\item $u_{\pi}(x_2^i)=6$;
	\item $u_{\pi}(x_3^i)=10$;
	\item $u_{\pi}(x_4^i)=11$;
	\item $u_{\pi}(x_5^i)=9$; and 
	\item $u_{\pi}(x_6^i)=1/2+1/2+10\frac{1}{4}=11\frac{1}{4}>11.$
	\end{itemize}

	Also $u_{\pi}(y^s)=3\times (10\frac{1}{4})=30\frac{3}{4}$ for all $s\in S'$ and $u_{\pi}(y^s)=0$ for all $s\in S\setminus S'$. We see that for each player, his utility is non-negative. Therefore there is no incentive for any player to deviate and form a singleton coalition. From Example~\ref{example:symm-core-empty} we also know that the only possible strongly blocking (and weakly blocking) coalition is $\{x_1^i\, x_5^i, x_6^i\}$ for any $i\in R$. However, $x_6^i$ has no incentive to be part $\{x_1^i,x_5^i,x_6^i\}$ because $u_{\pi}(x_6^i)=11$ and $v_{x_6^i}(x_5^i)+v_{x_6^i}(x_1^i)=6+5=11$. 
	Also $x_1^i$ and $x_5^i$ have no incentive to join $\pi(x_6^i)$ because their new utility will become negative because of the presence of the $y^s$ player. Assume for the sake of contradiction that $\pi$ is not core stable and $x_6^i$ can deviate with a lot of $x_6^j$s. But, $x_6^i$ can only deviate with a maximum of six other players of type $x_6^j$ because $i\in R$ is present in a maximum of three elements in $S$. In this case $x_6^i$ gets a maximum utility of only $1$. 
	Therefore $\pi$ is in the strict core (and thereby the core).

	We now assume that there exists a partition which is core stable. Then we prove that there exists an $S'\subseteq S$ such that $S'$ is a partition of $R$. 
	For any $s=\{k,l,m\}\in S$, the new utilities created due to the reduction gadget are only beneficial to $y^s$, $x_6^k$, $x_6^l$, and $x_6^m$.
	We already know that the only way the partition is core stable is if $x_6^i$ can be provided disincentive to deviate with $x_5^i$ and $x_1^i$.
	The claim is that each $x_6^i$ needs to be in a coalition with exactly one $y^s$ such that $i\in s\in S$ and exactly two other players $x_6^j$ and $x_6^k$ such that $\{i,j,k\}=s\in S$. We first show that $x_6^i$ needs to be with exactly one $y^s$ such that $i\in s\in S$. Player needs to be with at least one such $y^s$. If $x_6^i$ is only with other $x_6^j$s, then we know that $x_6^i$ gets a maximum utility of only $6\times 1/2=3$.  Also, player $x_6^i$ cannot be in a coalition with $y^s$ and $y^{s'}$ such that $i\in s$ and $i\in s'$ because both $y^s$ and $y^{s'}$ then get negative utility. Each $x_6^i$ also needs to be with at least 2 other players $x_6^j$ and $x_6^k$ where $j$ and $k$ are also members of $s$. If $x_6^i$ is with at least three players $x_6^j$, $x_6^k$ and $x_6^k$, then there is one element among $a\in \{j,k,l\}$ such that $a\notin s$. Therefore $y^s$ and $x_6^a$ hate each other and the coalition $\{y^s, x_6^i,x_6^j, x_6^k,x_6^k\}$ is not even individually rational. Therefore for the partition to be core stable each $x_6^i$ has to be with exactly one $y^s$ such that $i\in s$ and and least 2 other players $x_6^j$ and $x_6^k$ where $j$ and $k$ are also members of $s$. This implies that there exists an $S'\subseteq S$ such that $S'$ is a partition of $R$.\qed

	\end{proof}

	\section{Contractual strict core and Pareto optimality}
	\label{sec:csc}

	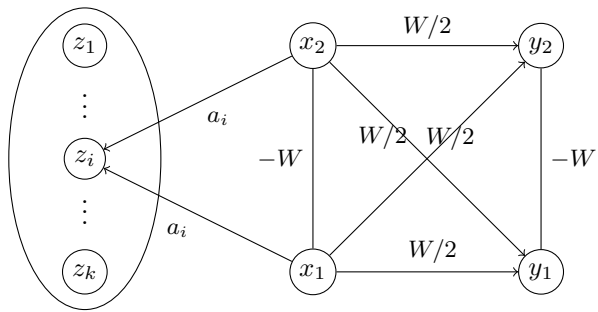
\begin{figure*}[htb]
	\centering
	\scalebox{1}{
	\begin{tikzpicture}[auto]
		\tikzstyle{player}=[draw, circle, fill=white, minimum size=11pt, inner sep=2pt]
		\tikzstyle{playerz}=[draw, circle, fill=white, minimum size=8pt, inner sep=2pt]
		\tikzstyle{lab}=[font=\small\itshape]
		\tikzstyle{dots}=[]

		\draw 
		(-3,3) node[playerz] (z1) {$z_1$}
		(-3,1.5) node[playerz] (zi) {$z_i$}
		(-3,0) node[playerz] (zk) {$z_k$}
		(0,0) node[player] (x1) {$x_1$}
			 ++(90:3cm) node[player] (x2) {$x_2$}
			 ++(0:3cm) node[player] (y2) {$y_2$}
			 ++(270:3cm) node[player] (y1) {$y_1$};
		\draw[->] (x1) to [] node[lab, above=6pt, left=-2pt] {$W/2$} (y1);	
		\draw[->] (x2) to [] node[lab, above=-1pt] {$W/2$} (y2);		
		\draw[-] (x1) to node[lab] {$-W$} (x2);
		\draw[->] (x1) to node[lab, above=20pt, right=-4pt] {$W/2$} (y2);
		\draw[->] (x2) to node[lab, below=22pt, right=-4pt] {$W/2$} (y1);
		\draw[-] (y2) to node[lab] {$-W$} (y1);
		\draw[->] (x1) to node[lab] {$a_i$} (zi);
		\draw[->] (x2) to node[lab] {$a_i$} (zi);

	\draw (-3,1.5) ellipse (1cm and 2cm);
		\node[dots] (subgame-dots) at (-3,0.9) {$\vdots$};
		\node[dots] (subgame-dots) at (-3,2.3) {$\vdots$};
	\end{tikzpicture}
	}
	\caption{Graphical representation of the \ASHG in the proof of Theorem~\ref{th:csc-hard}. For all $i\in \{1,\ldots, k\}$, an edge from $x_1$ and $x_2$ to $z_i$ has weight $a_i$. All other edges not shown in the figure have weight zero.}
	\label{fig:exampleCSC}
	\end{figure*}

	In this section, we prove that verifying whether a partition is CSC stable is coNP-complete. Interestingly, coNP-completeness holds even if the partition in question consists of the grand coalition.  
	The proof of Theorem~\ref{th:csc-hard} is by a reduction from the following weakly NP-complete problem.\\

	\noindent
	\textbf{Name}: {\sc Partition} \\
	\textbf{Instance}: A set of $k$ positive integer weights $A=\{a_1, \ldots, a_k \}$ such that $\sum_{a_i\in A}a_i=W$.\\
	\textbf{Question}: Is it possible to partition $A$, into two subsets $A_1\subseteq A$, $A_2\subseteq A$ so that $A_1\cap A_2=\emptyset$ and $A_1\cup A_2=A$ and $\sum_{a_i\in A_1}a_i=\sum_{a_i\in A_2}a_i=W/2$?\\

	\begin{theorem}\label{th:csc-hard}
	Verifying whether the partition consisting of the grand coalition is CSC stable is weakly coNP-complete.
	\end{theorem}
	\begin{proof}
	The problem is clearly in coNP because a partition $\pi'$ resulting by a CSC deviation from $\{N\}$ is a succinct certificate that $\{N\}$ is not CSC stable.
	We prove NP-hardness of deciding whether the grand coalition is \emph{not} CSC stable by a reduction from {\sc Partition}.
	We can reduce an instance of $I$ of {\sc Partition} to an instance $I'=((N,\pref),\pi)$ where $(N,\pref)$ is an \ASHG defined in the following way:

	\begin{itemize}
	\item $N=\{x_1,x_2,y_1,y_2, z_1,\ldots,z_k\}$,
	\item $v_{x_1}(y_1)=v_{x_1}(y_2)=v_{x_2}(y_1)=v_{x_2}(y_2)=W/2$,
	\item $v_{x_1}(z_i)=v_{x_2}(z_i)=a_i$, for all $i\in \{1,\ldots, k\}$
	\item $v_{x_1}(x_2)=v_{x_2}(x_1)=-W$,
	\item $v_{y_1}(y_2)=v_{y_2}(y_1)=-W$,
	\item $v_{a}(b)=0$ for any $a,b\in N$ for which $v_{a}(b)$ is not already defined, and
	\item $\pi=\{N\}$.
	\end{itemize}

	We see that $u_{\pi}(x_1)=u_{\pi}(x_1)=W$, $u_{\pi}(y_1)=u_{\pi}(y_2)=-W$, $u_{\pi}(z_i)=0$ for all $i\in \{1,\ldots, k\}$.
	We show that $\pi$ is not CSC stable if and only if $I$ is a `yes' instance of {\sc Partition}. 
	Assume $I$ is a `yes' instance of {\sc Partition} and there exists an $A_1\subseteq A$ such that $\sum_{a_i\in A_1}a_i=W/2$.
	Then, form the following partition

	$$\pi'=\{\{x_1,y_1\}\cup \{z_i \mid a_i\in A_1\}, \{x_2,y_2\}\cup \{z_i \mid a_i\in N\setminus A_1\} \} $$

	Then, 

	\begin{itemize}
		\item $u_{\pi'}(x_1)=u_{\pi'}(x_1)=W$; 
		\item $u_{\pi'}(y_1)=u_{\pi'}(y_2)=0$; and 
		\item $u_{\pi}(z_i)=0$ for all $i\in \{1,\ldots, k\}$. 
	\end{itemize}

	The coalition $C_1=\{x_1,y_1\}\cup \{z_i \mid a_i\in A_1\}$ can be considered as a coalition which leaves the grand coalition so that all players in $N$ do as well as before and at least one player in $C_1$, i.e., $y_1$ gets strictly more utility. Also, the departure of $C_1$ does not make any player in $N\setminus C_1$ worse off.

	Assume that $I$ is a `no' instance of {\sc Partition} and there exists no $A_1\subseteq A$ such that $\sum_{a_i\in A_1}a_i=W/2$. We show that no CSC deviation is possible from $\pi$. 
	We consider different possibilities for a CSC blocking coalition $C$:

	\begin{enumerate}
	\item $x_1,x_2, y_1, y_2 \notin C$, 
	\item $x_1,x_2\notin C$ and there exists $y\in \{y_1,y_2\}$ such that $y\in C$,
	\item $x_1,x_2,y_1,y_2\in C$,
	\item $x_1,x_2\in C$ and $|C\cap \{y_1,y_2\}|\leq 1$,
	\item there exists $x\in \{x_1,x_2\}$ and $y\in \{y_1,y_2\}$ such that $x,y\in C$, $\{x_1,x_2\}\setminus x \nsubseteq C$, and $\{y_1,y_2\}\setminus y \nsubseteq C$
	\end{enumerate}

	We show that in each of the cases, $C$ is a not a valid CSC blocking coalition.

	\begin{enumerate}
	\item If $C$ is empty, then there exists no CSC blocking coalition. If $C$ is not empty, then $x_1$ and $x_2$ gets strictly less utility when a subset of $\{z_1,\ldots, z_k\}$ deviates.
	\item In this case, both $x_1$ and $x_2$ gets strictly less utility when $y\in \{y_1,y_2\}$ leaves $N$.
	\item If $\{z_1,\ldots, z_k\}\subset C$, then there is no deviation as $C=N$. If there exists a $z_i\in \{z_1,\ldots, z_k\}$ such that $z_i\notin C$, then $x_1$ and $x_2$ get strictly less utility than in $N$.
	\item If $|C\cap \{y_1,y_2\}|= 0$, then the utility of no player increases. If $|C\cap \{y_1,y_2\}|=1$, then the utility of $y_1$ and $y_2$ increases but the utility of $x_1$ and $x_2$ decreases.
	\item Consider $C=\{x,y\}\cup S$ where $S\subseteq \{z_1,\ldots, z_k\}$. Without loss of generality, we can assume that $x=x_1$ and $y=y_1$. We know that $y_1$ and $y_2$ gets strictly more utility because they are now in different coalitions. Since $I$ is a `no' instance of {\sc Partition}, we know that there exists no $S$ such that $\sum_{a\in S}v_{x_1}(a)=W/2$. If $\sum_{a\in S}v_{x_1}(a)>W/2$, then $u_{\pi}(x_2)<W$. If $\sum_{a\in S}v_{x_1}(a)<W/2$, then $u_{\pi}(x_1)<W$. 
	\end{enumerate}

	Thus, if $I'$ is a `no' instance of {\sc Partition}, then there exists no CSC deviation. \qed
	\end{proof}

	From the proof of Theorem~\ref{th:csc-hard}, it can be seen that $\pi$ is not Pareto optimal if and only if $I$ is a `yes' instance of {\sc Partition}. 

	\begin{theorem}\label{th:GC-PO}
	Verifying whether the partition consisting of the grand coalition is Pareto optimal is coNP-complete.
	\end{theorem}

	\section{Conclusion and Discussion}

	We presented a number of new computational results concerning stable partitions of \ASHGS. 
	First, we proposed a polynomial-time algorithm for computing a contractually individually stable (CIS) partition. Secondly, we showed that checking whether the core or strict core exists is NP-hard in the strong sense, even if the preferences of the players are symmetric. Finally, we presented the first complexity result concerning the contractual strict core (CSC), namely that verifying whether a partition is in the CSC is coNP-complete. We saw that considering CSC deviations helps reason about the more complex Pareto optimal improvements. As a result, we established that checking whether the partition consisting of the grand coalition is Pareto optimal is also coNP-complete.

	We note that Algorithm~\ref{alg-CIS-general} may very well return a partition that fails to satisfy individual rationality, \ie players may get negative utility. It is an open question how to efficiently compute a CIS partition that is guaranteed to satisfy individual rationality.
	We also note that Theorem~\ref{th:csc-hard} may not imply anything about the complexity of \emph{computing} a CSC partition. Studying the complexity of computing a CSC stable partition is left as future work.


\begin{contact}
Haris Aziz, Felix Brandt, and Hans Georg Seedig\\
Department of Informatics\\ 
Technische Universit\"at M\"unchen\\
85748 Garching bei M\"unchen, Germany\\
\texttt{\small\{aziz,brandtf,seedigh\}@in.tum.de}
\end{contact}

\end{document}